\documentclass[11pt]{article}

\newenvironment{hangref}
  {\begin{list}{}{\setlength{\itemsep}{4pt}
  \setlength{\parsep}{0pt}\setlength{\leftmargin}{+\parindent}
  \setlength{\itemindent}{-\parindent}}}{\end{list}}

\def\qed{\relax\ifmmode\hskip2em \fbox{ }\else\unskip\nobreak\hskip1em 
$\fbox{}$\fi}

\newsavebox{\theorembox}
\newsavebox{\lemmabox}
\newsavebox{\corollarybox}
\newsavebox{\propositionbox}
\newsavebox{\examplebox}
\newsavebox{\propertybox}
\savebox{\theorembox}{\bf Theorem}
\savebox{\lemmabox}{\bf Lemma}
\savebox{\corollarybox}{\bf Corollary}
\savebox{\propositionbox}{\bf Proposition}
\savebox{\examplebox}{\bf Example}
\savebox{\propertybox}{\bf Property}
\newtheorem{theorem}{\usebox{\theorembox}}
\newtheorem{lemma}[theorem]{\usebox{\lemmabox}}

\newtheorem{definition}{{\sc Definition}\rm }[section]
\newtheorem{definitions}[definition]{{\sc Definitions\rm }}

\newenvironment{proof}{{\noindent\bf Proof~}}{\(\qed\)\vspace*{\proofskip} }
\newlength{\proofskip}
\begin{document}

\begin{center}

{\LARGE 
Fast Algorithms for Exact String Matching
}


\footnotesize

\mbox{\large Srikrishnan Divakaran}\\
DA-IICT, 
Gandhinagar, Gujarat, India 382007,
\mbox{
Srikrishnan\_divakaran@daiict.ac.in }\\[6pt]
\normalsize
\end{center}

\baselineskip 20pt plus .3pt minus .1pt


\noindent 
\begin{abstract}
Given a pattern string $P$ of length $n$ and a query string $T$ of length $m$, where     the characters of $P$ and $T$  are drawn 
from an alphabet          of size $\Delta$,    the {\em exact string matching} problem consists of finding all occurrences of $P$ 
in $T$. For this problem, we present algorithms      that in $O(n\Delta^2)$ time preprocess $P$ to essentially identify $sparse(P)$, a  rarely occurring substring of $P$, and  then use it to find occurrences of $P$ in $T$ efficiently. Our algorithms require  a  worst      case search time of $O(m)$, and expected search time of $O(m/min(|sparse(P)|, \Delta))$,
where $|sparse(P)|$ is atleast $\delta$ (i.e. the number of distinct characters in $P$), and for most pattern strings it is 
observed to be $\Omega(n^{1/2})$. 
\end{abstract}
\noindent\hrulefill
\section{Introduction}
Given a pattern string $P$ of length $n$ and a query string $T$ of length $m$, where the characters of $P$ and $T$  are drawn 
from an alphabet          of size $\Delta$,    the {\em exact string matching} problem consists of finding all occurrences of 
$P$ in $T$. This is a  fundamental  problem  with wide range of applications in Computer Science (used    in   parsers,  word
processors, operating systems,  web  search engines, image processing and natural language processing),    Bioinformatics and 
Computational Biology (Sequence   Alignment and Database Searches). The algorithms  for   exact   string matching   can    be 
broadly categorized into the following categories: (1) character            based  comparison algorithms, (2) automata  based 
algorithms, (3) algorithms based on bit-parallelism and (4) constant-space algorithms. In   this  paper, our      focus is on 
designing    efficient         character based comparison algorithms for exact string matching. For a comprehensive survey of 
all categories of exact string matching algorithms, we refer the readers to Baeza-Yates[17], Gusfield[25], Charras et al [28]
, Chochemore et al [29] and Faro et al [30]. 
\newline \newline
A Typical character based comparison algorithm  can    be     described within the following  general   framework as follows:
\begin{itemize} 
\item [(1)]
First, initialize the search window to be the first $n$ characters of the query string $T$ (i.e. align the $n$ characters  of 
the pattern string $P$ with the  first $n$ characters of $T$). 
\item [(2)] 
Repeat the following until the search      window is  no  longer contained within the query string $T$:
\begin{quote} 
inspect  the aligned pairs in some order until there is either a mis-match in an  aligned pair or there is   a complete match 
among all the    $n$ aligned pairs. Then       shift the search window to the right. The order in which the aligned pairs are 
inspected and the length by  which the search window is shifted differs from one algorithm to another.
\end{quote} 
\end{itemize} 
The mechanism that the above framework provides is usually referred to as the {\em sliding window mechanism} [30, 31]. 
The algorithms that employ the sliding window mechanism can be further  classified  based on the order in which they inspect 
the  aligned pairs into the following broad categories: (1) left to right scan; (2) right to left scan; (3) scan in specific 
order, and (4) scan in   random order or scan order is not relevant. The algorithms that inspect the aligned pairs from left 
to right are the most natural algorithms; the algorithms that inspect the aligned pairs from right to left generally perform 
well in practice; the algorithms that inspect the aligned pairs  in a  specific  order  yield  the  best theoretical bounds. 
For    a     comprehensive    description   of the exact string matching algorithms and access to an excellent framework for 
development, testing and analysis of exact string matching algorithms, we refer the readers to the {\em smart tool}  (string 
matching research tool) of Faro and T. Lecroq[31]. \newline \newline 
For algorithms that inspect aligned pairs from left to right, Morris and Pratt [1] proposed the first known       linear time 
algorithm. This algorithm was improved by Knuth, Morris and Pratt [4] and        requires $O(n)$    preprocessing time and a 
worst case search time of at most $2m-1$ comparisons. For small     pattern strings and reasonable probabilistic assumptions 
about the distribution of characters in the query string, hashing [2,9] provides an $O(n)$ preprocessing  time and $O(m)$ worst 
case search time solution. For pattern strings  that fit within a  word   of    main memory, Shift-Or [16,21] requires $O(n + 
\Delta)$ preprocessing time and a search time of $O(m)$ and  can     also be easily adapted  to solve     approximate string 
matching problems. 
For algorithms that inspects aligned pairs from right to left, The Boyer-Moore [3] algorithm is one of the classic algorithms 
that requires $O(n +\Delta)$ preprocessing time and a worst case search time of $O(nm)$ but in practice is  very fast. There 
are   several variants that simplify the Boyer-Moore algorithm and mostly  avoid its quadratic behaviour. Among the variants 
of Boyer-Moore, the    algorithms of   Apostolico and Giancarlo [7,24],  Crochemore et al [13, 23] (Turbo BM), and Colussi (Reverse Colussi) [12, 22]  
have $O(m)$ worst case search time and are efficient in minimizing the number of character comparisons, whereas the   Quick 
Search [10], Reverse Factor [19], Turbo Reverse Factor [24], Zhu and Takaoka [8] and Berry-Ravindran [27] algorithms are very efficient in practice.
For Algorithms that inspects the  aligned pairs in a specific order, Two Way algorithm [13], Colussi [12], Optimal Mismatch 
and Maximal Shift [10], Galil-Giancarlo [18],    Skip Search , KMP Skip Search and Alpha Skip Search [26] are some of the well known algorithms.
Two way algorithm was the first known linear time optimal space algorithm. The     Colussi      algorithm      improves the 
Knuth-Morris-Pratt algorithm and requires at most 3/2n text character comparisons in the worst case. The    Galil-Giancarlo 
algorithm improves the Colussi algorithm in one special case which enables it to perform at most        4/3n text character
comparisons in the worst case.  
For Algorithms that inspects the aligned pairs in any order,the Horspool [5], Quick Search [10], Tuned Boyer-Moore [14], Smith [15] and Raita [20]
algorithms are some of the well known algorithms. All these algorithms have   worst case search time that is quadratic but 
are known to perform well in practice. 
\newline \newline 
{\bf Our Results}: In this paper, we present two similar algorithms $A$ and $B$ for exact     string    matching. Both these 
algorithms   employ   sliding     window mechanism, preprocess $P$ in $O(n\Delta^2)$ time to essentially identify $sparse(P)
$, a rarely occurring substring of $P$ characterized by two characters of $P$, and then   use it to find  occurrences of $P$ 
in $T$ efficiently.
Algorithms   $A$      and $B$ have worst case search times of $O(m)$ and $O(mn)$ respectively. However, both of them have an 
expected search time  of $O(m/min(|sparse(P)|,\Delta))$. The main difference between these two algorithms is that  Algorithm $A$   inspects   the   aligned pairs in  the  search  window  in the order specified by Apostolico-Giancarlo's [7] Algorithm, whereas Algorithm $B$ inspects the aligned pairs in random order. This makes algorithm $B$ much simpler than $A$ and equally effective in practice. In  terms
of     preliminary empirical analysis, for most pattern strings $P$, we observe that $|sparse(P)|$ is  $\Omega(n^{1/2})$.
We also believe that a tighter analysis of our algorithms can result in sub-linear worst case run-time from the perspective of
randomized analysis. \newline \newline 
The  rest    of this paper   is    organized as follows: In Section $2$, we present  our     Algorithms 
$A$ and $B$. In Section $3$, we present the analysis of these  algorithms, and in Section $4$   we present   our conclusions 
and future work.
\section{Algorithms for Exact String Matching}
In this section, we present two   similar     Algorithms $A$ and $B$, that given a pattern string $P$ and a query string $T$, finds   all occurrences of $P$ within $T$. First,we introduce some definitions that are essential for defining our Algorithms 
$A$ and $B$. Then, we present Algorithms $A$ and $B$.
\begin{definitions}
Given a pattern string $P$ of length $n$ and a query string $T$ of length $m$, we define $N_{i}(P)$, $i \in [1..n]$, denote the length of the longest suffix  of  $P[1..i]$ that matches a suffix of $P$, and $M$ to be a  $m$ length vector whose   jth entry $M[j] = k$ indicates that a          suffix of $P$ of length at least $k$ occurs in $T$ and ends at position $j$. (See Apostolico-Giancarlo Algorithm [7]). 
\end{definitions}
\begin{definitions}
Given a pattern string $P$, and an ordered   pair of characters $u, v \in \Sigma$ (not necessarily distinct), we define $sparse^{(u,v)}(P)$, the $2$-sparse pattern of $P$ with respect to $u$ and $v$, to be the rightmost     occurrence of a substring of $P$ of longest length that starts with $u$, ends with $v$, but does not contain $u$ or $v$      within it. We define $sparse(P)$  to   be the longest among the $2$-sparse patterns of $P$.
\end{definitions}
\begin{definitions}
Given $sparse(P)$, the   longest   $2$-sparse pattern of $P$, we define $startc(P)$ and $endc(P)$ to be  the respective first and last characters of $sparse(P)$, and   $startpos(P)$    and    $endpos(P)$  be the respective   indices of the first and last characters of $sparse(P)$ in $P$. For $c \in \Sigma$, if $c \in sparse(P)$,       $shift^{c}(P)$ is  the distance between  the rightmost occurrence of $c$ in $sparse(P)$ and the  last character of $sparse(P)$.  If $c$ is not present in $P$ then $shift^{c}(P)$  is set to $n$, the length of $P$. If $c$ is present in $P$ but not in $sparse(P)$ then $shift^{c}(P)$ is set to $|sparse(P)|+1$. \end{definitions}
{\bf BASIC IDEA}:  First, we preprocess $P$ to identify $sparse(P)$, {\it a rarely occurring substring of $P$ characterized 
by two characters in $P$}, and compute statistics of its occurrence relative to   other characters in $P$. Then, during the 
search   phase, we set the search window to be the first $n$ characters of $T$ (i.e. align the $n$ characters of $P$ to the 
first $n$ characters of $T$). Then, we do the following repeatedly until the search window reaches the end of $T$:
\begin{quote}
Check whether there is a match between the first and last characters of $sparse(P)$ and their respective  aligned characters in the search window. In the case    of a    match, Algorithm $A$ (Algorithm $B$) invokes $Apostolico-Giancarlo$ Algorithm ($Random-Match$) to     look    for      exact     match between $P$ and the characters in the search window. If there is an exact match then it reports the match, shifts the search window by at least the length of $sparse(P)$ and   then  continues. However, if   there   is      a mismatch in either the first or last  character of $sparse(P)$ or during the invocation of Apostolico-Giancarlo (Random-Match), then it shifts the search window based on the  statistics of $sparse(P)$'s occurrence relative to the mismatched character and   then continues. 
\end{quote}
{\bf DESCRIPTION OF ALGORITHM $A$}: We first describe the preprocessing and search phases of Algorithm $A$. Then, we 
present Algorithm $A$ formally. \newline \newline 
{\bf Preprocessing Phase}: In the pre-processing phase, we first compute $N_{i}(P)$, for $i \in [1..n]$, where $N_{i}(P)$   is    the   longest suffix of $P[1.i]$ that matches a suffix of $P$. Second, for each character $c \in \Delta$, we determine 
the list of indices in $P$ where it occurs, and from this we determine $sparse(P)$, the longest $2$-sparse pattern of $P$, as follows:
\begin{itemize} 
\item[(a)] For each ordered  pair of  characters $a, b \in P$ (not necessarily distinct), determine $sparse^{(a,b)}(P)$,the $2$-sparse pattern with respect  to  $a$  and $b$;
\item[(b)] Choose the longest among the $2$-sparse patterns determined in $(a)$.
\end{itemize} 
Third, from $sparse(P)$,  we determine    $startc(P)$ and $endc(P)$, the respective first and last  characters of $sparse(P)$, and   $startpos(P)$   and  $endpos(P)$, the respective indices within $P$ of the first and last  characters of $sparse(P)$. Finally, for $c \in \Sigma$, determine $shift^{c}(P)$, the distance   between  the rightmost occurrence of  $c$ in $sparse(P)$ and the  last character of $sparse(P)$. \newline 
{\bf Search Phase}: In the search phase, we first set the search window to be the      first $n$ characters of $T$ (i.e. align the $n$ characters of $P$ to the first $n$ characters of $T$). Then, while   the search window is contained within the query string $T$,  compare $endc(P)$ and $startc(P)$  (i.e the last and first characters of $sparse(P)$)  with the respective characters at offset $endpos(P)$ and  $startpos(P)$  within the search window. The following three dis-joint scenarios (events) are possible:
\begin{itemize}
\item[(i)] [Type-1 event] If the  character $c$ at offset $endpos(P)$ within the search window is not $endc(P)$ then we shift  the search window to the right by $shift^{c}(P)$;
\item[(ii)] [Type-2 event] If   the   character  $c$ at  offset $endpos(P)$ within the search window is $endc(P)$ but the character $d$ at offset   $startpos(P)$ is  not $startc(P)$ then we shift the search window to    the      right by $|sparse(P)|+1$($|sparse(P)|$) when  $c$ and $d$ are different characters ($c$ and $d$ are the same characters). 
\item[(iii)] [Type-3 event] If          the character at offset $startpos(P)$ within the search window is $startc(P)$ and at offset $endpos(P)$ is $endc(P)$ then we use Apostolico-GianCarlo's  Algorithm to look for an occurrence    of $P$ within the search window in $T$ and then  shift the search window to the right by    $|sparse(P)|+1$($|sparse(P)|$) when  $c$ and $d$ are different characters ($c$ and $d$ are the same characters).
\end{itemize}
{\bf ALGORITHM $A$}
\begin{tabbing}
Input(s): \= (1) Pattern string $P$ of length $n$; \\
          \> (2) Query string $T$ of length $m$;\\
Output(s): The starting positions of the occurrences of $P$ in $T$; \\ 
{\em Preprocessing}: \= (1) For $i \in [1..n]$, compute $N_{i}(P)$ using the $Z$ Algorithm [25]. \\
					 \> (2) \= For $c \in \Delta$, determine $List^{c}(P)$, the list of positions in $P$ where $c$ occurs. \\
                     \>     \> Then \= from these lists, compute \\ 
                     \>     \>      \= [a] for each $u, v \in \Delta$, $sparse^{(u,v)}(P)$; \\
                     \>     \>      \> [b] $sparse(P) = |sparse^{(a,b)}(P)|=max_{u,v \in \Delta}^{} |sparse^{(u,v)}(P)|$. \\
                     \> (3) From $sparse(P)$, compute $startc(P), endc(P), startpos(P)$, and $endpos(P)$. \\
                     \> (4) For $c \in \Sigma$, compute $shift^{c}(P)$. \\                     
{\em Search}:\= \\
          \> [1] \= [a] Set $i=0$ and $j=n$ [Search Window set to $[1..n]$] \\
          \>     \> [b] Set $\hat{j}= endpos(P)$ and $\hat{i} = startpos(P)$; \= [Indices of last and first characters \\
          \>     \>                                                           \>  of $sparse(P)$ in $P$] \\
          \> [2] \= while $(j  <  m)$ [While the search window is contained in $T$] \\
          \>     \> [a] Let $c = T[i+\hat{j}]$; $d=T[i+\hat{i}]$; \= [Characters in search window aligned with \\
          \>     \>                                               \>  the last and first characters of $sparse(P)$]\\
          \>     \> [i] \= if  \= ($c \ne endc(P)$)  [Type-1 event] \\
          \>     \>     \>     \>  $i = i + shift^{c}(P)$; $j = j + shift^{c}(P)$; [Shift search window by $shift^{c}(P)$] \\
          \>     \> [ii] \= else\=if ($d \ne startc(P)$) [Type-2 event] \\
          \>     \>     \>     \> if (\=$startc(P) == endc(P)$) \\
           \>    \>     \>     \>     \> $i=i+|sparse(P)|$; $j=j+|sparse(P)|$ [Shift search window by $|sparse(P)|$] \\
           \>     \>     \>     \> else\= \\ 
           \>     \>     \>     \>     \> $i=i+|sparse(P)|+1$; $j=j+|sparse(P)|+1$ [Shift search window by $|sparse(P)|+1$]\\
          \>     \> [iii]\= else\= [Type-3 event] \\
          \>     \>     \>     \> Call $Apostolico-GianCarlo(T, P, i, j)$; [Look for $P$ in $T[i+1..j]$]\\
          \>     \>     \>     \> if (\=$startc(P) == endc(P)$) \\
          \>     \>     \>     \>     \> $i=i+|sparse(P)|$; $j=j+|sparse(P)|$ [Shift search window by $|sparse(P)|$] \\
          \>     \>     \>     \> else\= \\ 
          \>     \>     \>     \>     \> $i=i+|sparse(P)|+1$; $j=j+|sparse(P)|+1$.[Shift search window by $|sparse(P)|+1$] \\
\end{tabbing}
{\bf ALGORITHM $B$} \newline 
We now define Algorithm $B$  by making the following simple modification to Step $[2][a][iii]$ of  the search phase of Algorithm $A$: 
\begin{quote} 
Replace the statement "Call $Apostolico-GianCarlo(T, P, i, j)$" by          the     statement "Call $Random-Match(T, P, i, j)$".
\end{quote} 
The $Apostolico-Giancarlo$ Algorithm determines an     exact match between      $P$ and $n$ characters in the search window by inspecting      the aligned pairs     in a specific order until it encounters a mismatch or finds a match in all $n$ characters. However, $Random-Match$    inspects    the        aligned      pairs in a random order until it encounters a mismatch or finds a match in all $n$ characters. 
\section{Analysis of Algorithms $A$ and $B$}
In this section, we present the analysis of Algorithms $A$ and $B$. We now present the main results in this paper. The  proofs follow.
\begin{theorem}
Given any pattern  string  $P$  of length $n$ and a query string $T$ of length $m$, Algorithm $A$ finds all occurrences of $P$ in $T$ is $O(m)$ time.
\end{theorem}
\begin{theorem}
Given      any pattern string $P$ of length $n$ and a query string $T$ of length $m$ where each character is drawn uniformly at random, Algorithms $A$ and $B$ find all occurrences of $P$ in $T$ in $O(m/min(|Sparse(P)|, \Delta))$ expected time, where  $|sparse(P)|$ is atleast $\delta$ (i.e the number of distinct characters in $P$).
\end{theorem}
\begin{proof} 
{\bf of Theorem $1$}: The Algorithm $A$ during its search phase searches for $P$ by essentially looking for a match for the first and last characters of      $sparse(P)$, a $2$-sparse pattern of $P$, with the characters in the search window   at offsets $startpos(P)$ and $endpos(P)$ respectively. Let  $c = endc(P)$ and $d =startc(P)$ be   the   last  and first  characters of $sparse(P)$. The following three scenarious(events) are possible. (i){\em Type-1} event  happens if there  is a mis-match between $c$ and the character at offset $endpos(P)$ within  the search window; (ii) {\em Type-2} event happens if there is  a match between $c$  and the  character   at offset $endpos(P)$ within the search  window and a mis-match between  $d$  and  the character at offset $startpos(P)$ within the search window, and (iii) {\em Type-3} event happens when there is a match between $c$ and the character  at  offset  $endpos(P)$ within  the search window and a match between $d$ and the character  at offset $startpos(P)$ within the search window. \newline \newline
In the case of Type-1 event, there is exactly one character in $T$ that is looked at and the search window   is     shifted by $shift^{c}(P) \ge 1$. In the case of Type-2 event, there   are two characters in $T$ that are looked at and the search window is shifted by $|sparse(P)|$ $(|sparse(P)+1)$ when $c=d$ $(c \ne d)$. In the case of type $3$ event, we call the Apostolico-Giancarlo  Algorithm where the query string is the $n$ characters in the  current search window. However, we maintain the $N$ and $M$  vectors as global variables so that when we repeatedly invoke the    Apostolico-Giancarlo Algorithm the M values computed for any particular  position of $T$ during any invocation is available without recomputation for future invocations. This ensures  that the total number of character comparisons done during type-3 events is $O(m)$. This bound follows from the analysis of Apostolico-Giancarlo Algorithm. Now to bound the total number of comparisons done by algorithm $A$, we only need  to compute the number of comparisons performed by $A$ that are associated with $Type-1$ and $Type-2$ events. \newline \newline
We bound the number of comparisons  made by $A$ due to  Type-1 and     Type-2 events   by     partitioning the search phase into  sub-phases, where each sub-phase consists of maximal sequence of events that begins with any type of event and  is terminated by a  Type-3 event, and then account for the number of comparisons made during Type-1 and   Type-2 events of the sub-phases. Notice that except the first sub-phase, every other sub-phase begins  with either a Type-1 or a          Type-2 event,  and ends with a    contiguous run of one or more Type-3 events. Notice that across sub-phases there can be a   overlap in the character comparisons only between the   last event of a sub-phase and the first   event of the next sub-phase. From an earlier observation, we notice that each Type-1 (Type-2) event   requires 1 (2) character    comparisons and the search window is shifted by at least $1$. This implies   that the total number of character comparisons associated with type-1 and type-2 events is at most $2m$. Therefore the total number of character comparisons due to Type-1, Type-2 and Type-3 events is $O(m)$. \newline \newline
Now, we   establish that Algorithm $A$ finds all occurrences of $P$ in $T$. In the case of Type-1 event, we know that     $c \ne endc(P)$ and the search window is shifted to the right by $shift^{c}(P)$. Recall  that $shift^{c}(P)$ indicates the    number of positions to the left of $endc(P)$ in $P$  where the earliest occurrence of  $c$ happens. Now, by  shifting the search window by $shift^{c}(P)$, we  will show that no occurrence of $sparse(P)$ will be skipped  and hence no occurrence of $P$ will be skipped. There are two     situations     possible      depending      on whether or not $endc(P)$ occurs within the shifted interval. If $endc(P)$ did      not occur within the shifted interval of $T$ then we can see that  shifting the search window to the right by $shift^{c}(P)$ will not      result in skipping $sparse(P)$. Now, we consider the situation when $endc(P)$ occurs     within the shifted interval.  From  the   definition of $shift^{c}(P)$, we can see that earliest occurrence of $c$ in $P$ will be $shift^{c}(P)$ positions to  the left of $endpos(P)$, whereas in this situation $c$ occurs in the search window less than $shift^{c}(P)$ positions to the left of $endpos(P)$. Therefore, we can conclude that no occurrence of $sparse(P)$ can start within the shifted portion of the search window. In the case of $Type-2$ event, we can observe that $c == endc(P)$ but $d \ne startc(P)$ and the search window is shifted  to the right by either $|sparse(P)|$ (|sparse(P)|+1) if $c=d$ ($c\ne d$). Notice  in       this case, since characters $c$     and $d$       do not occur within $sparse(P)$ (i.e. occur only at the start and end of $sparse(P)$), the next occurrence of $sparse(P)$ in $T$ cannot start within the shifted portion of the search window. In    the case of $Type-3$ event, we can observe that $c ==endc(P)$ and $d = startc(P)$ and after invocing Apostolico- Giancarlo Algorithm we   shift    the window to the right by either $|sparse(P)|$ ($|sparse(P)| + 1$) if $c=d$ ($c \ne d$). Notice  in     this case also, since characters $c$and $d$ do not occur within $sparse(P)$, the    next   occurrence of $sparse(P)$ cannot start within the shifted portion of the search window. Hence, Algorithm $A$ finds all   occurrences of $P$ in $T$ correctly. 
\end{proof} \newline \newline 
\begin{proof} 
{\bf of Theorem $2$}: We bound the expected number of comparisons performed by Algorithm $A$ during its search phase  by looking at the expected number of comparisons performed in comparison to the expected length by which the   search window is shifted during each of the three type of events it encounters. For Type-1 and Type-2 events the number of character comparisons with the query string $T$ is at most $2$. For Type-3 events, we are invoking the Apostolico-Giancarlo Algorithm. From Lemma $6$, we know that when invoking Apostolico-Giancarlo Algorithm for the pattern string $P$ and the query string $T$ whose characters are drawn uniformly at random, the expected number of matches before a mismatch is $O(1)$. Therefore the number of character comparisons for    Type-3 event is also $O(1)$. Now to bound the total number of comparisons, it is sufficient to bound the number of events.  From   Lemma $5$, we know that the expected  length by which the search window is shifted after encountering a Type-1,  Type-2 or Type-3 event is at least $O(min(|sparse(P)|,\Delta)$. Since the total amount by which the search window can be shifted is $m$, we can therefore see that the total number of events is $O(m/min(|sparse(P)|,\Delta))$. The proof for Algorithm $B$ is almost the same.
\end{proof} 
\begin{lemma}
For any    pattern string $P$, the length of $sparse(P)$, the longest $2$-sparse  pattern of $P$, is at least $\delta$,   where $\delta$ is the number of distinct characters in $P$.
\end{lemma}
\begin{proof}
From definition, we know that $P$ has $\delta$ distinct characters. Now, let $a$ be the character in $P$ whose last   occurrence has the smallest index and let its index in $P$ be denoted by $start$. Let $b$ be the character in $P$ whose first occurrence in $P$ to the right of $start$ has the highest index and let its position in $P$ be denoted by $end$. Now, we can observe that every character in $P$ appears at least once within the interval $[start, end]$ and characters $a$ and $b$ do not appear within the interval $[start, end]$.  Since  there are $\delta$ distinct characters in $P$, the length of $sparse^{(a,b)}(P)$ is at least $\delta$. Hence the length of $sparse(P)$ is at least $\delta$.
\end{proof}
\begin{lemma}
For any pattern string $P$, Algorithm $A$ preprocesses $P$ in $O(n \Delta^{2})$ time to determine (i)$N_i(P)$, for $i \in [1..n]$, (ii) $sparse(P)$, and (iii) $shift^{c}(P)$, for $c \in \Sigma$, where $\Delta$ is the number of characters in its alphabet $\Sigma$.
\end{lemma}
\begin{proof}
First, we would like to recall that $N_{i}(P)$, $i \in [1..n]$, is the length of the longest suffix of $P[1..i]$ that is also  a suffix of $P$. We compute compute $N_i(P)$, for $i \in [1..n]$ in $O(n)$ time using the $Z$ Algorithm [].     Second,   for each ordered pair $a,b \in P$ of characters, we can scan $P$ in $O(n)$ time to find the maximum length substring  of $P$ starting with 
$a$ and ending with $b$ such that there is no occurrence of $a$ or $b$ in between. Since   there are     $\Delta^2$ ordered pairs, we can        trivially find the maximum length for all pairs of characters in $P$ and from them   choose the longest in $O(n\Delta^2)$ time. Finally, we would like to recall that $shift^{c}(P)$, for $c\in \Sigma$, is the distance between the last character of 
$sparse(P)$  and the rightmost occurrence of $c$ in $sparse(P)$. If $c$ is not present in $P$    then $shift^{c}(P)$  is set to $n$, the length of the pattern string $P$. Therefore, in $O(n)$ time, we can scan
$sparse(P)$ to find $shift^{c}(P)$, for $c \in \Delta$. 
\end{proof}
\begin{lemma}
For   any    pattern   string $P$, during the search phase of Algorithm $A$, the expected length of shift of $P$ after a Type-1, Type-2 or Type-3 event is at least $O(min(|sparse(P)|,\Delta))$.
\end{lemma}
\begin{proof}
Notice that the query string $T$ is of length $m$ and each of its characters are drawn uniformly    at random from the alphabet $\Sigma$ of size $\Delta$. In the case of a Type-1 event, we can observe that the search window  is shifted by $shift^{c}(P)$.    Notice       that the mismatch character $c$ is equally likely to be any character in $\Sigma$ other than $endc(P)$. So, we can observe that if $c \in sparse(P)$, then $shift^{c}(P)$ is equally likely to be any value in the interval $[1..\delta]$ and if $c \notin sparse(P)$ then $shift^{c}(P)$ is at least $|sparse(P)|$. Therefore, the expected shift value will be at least $(1 + 2 + ... + \delta) + (\Delta-/\delta)|sparse(P)|$. Now, based on whether $\delta < \Delta/2$ or $\delta \geq \Delta/2$, we can evaluate the above sum. If $\delta < \Delta/2$, we can observe that this sum is $O(|sparse(P)|)$, otherwise the sum is $O(|\Delta|)$. Therefore, the above sum is at least $O(min(|sparse(P)|,\Delta))$. In the case of Type-2 event, we can observe that the search window is shifted by the length of    $sparse(P)$. Similarly after a type-3 event, Apostolico-Giancarlo Algorithm is called and after that the  search window is shifted by the length of $sparse(P)$.  Hence in all three types of events the expected search window shift is at least $O(min(|sparse(P)|,\Delta)) \geq \delta$. Hence the result.
\end{proof}
\begin{lemma}
For any given pattern string $P$ of length $n$ from $\Sigma$    and a query string $T$ whose characters are drawn independently and uniformly from $\Sigma$, the expected number of matches before a mismatch when invoking Apostolico-Giancarlo or Random-Match Algorithm  is $O(1)$.
\end{lemma}
\begin{proof}
In $A$, each time we invoke Apostolico-Giancarlo/Random-Match  Algorithm, we   attempt to match $P$ with the search window consisting of $n$ length substring of $T$, where each character is drawn independently and uniformly from $\Sigma$. So, the  expected length of a match = $(1/\Delta + 1/\Delta^2 + 1/\Delta^3 + ... )(\Delta-1/\Delta) = O(1)$.
\end{proof}
\section{Conclusions and Future Work}
In this paper, we present algorithms for exact string matching that require  a  worst      case search time of $O(m)$, and sub-linear expected search time of $O(m/min(|sparse(P)|, \Delta))$, where $|sparse(P)|$ is at least $\delta$ (i.e. the number of distinct characters in $P$), and for most pattern strings is observed to be $\Omega(n^{1/2})$. We believe that a tighter analysis of our algorithms can establish sub-linear worst case run-time from the perspective of randomized analysis. We also believe that for a large class of pattern strings it seems plausible that one can theoretically establish that $|sparse(P)| = \Omega(n^{1/2})$. \newline \newline 
\section*{References}
\begin{hangref} 
\item[1] 	J. H. Morris and V. R. Pratt. A linear pattern-matching algorithm. n.40, (1970).	
\item[2]    M.C. Harrison. Implementation of substring test by hashing, Commun. ACM, vol 14, n.2, pp. 777-779, (1971)
\item[3]  	R. S. Boyer and J. S. Moore. A fast string searching algorithm. Commun. ACM, vol.20, n.10, pp.762--772, (1977.	
\item[4]  	D. E. Knuth and J. H. Morris and V. R. Pratt. Fast pattern matching in strings. SIAM J. Comput., vol.6, n.1, pp.323--350, (1977).	
\item[5]  	R. N. Horspool. Practical fast searching in strings. Softw. Pract. Exp., vol.10, n.6, pp.501--506, (1980).	
\item[6]  	Z. Galil and J. Seiferas. Time-space optimal string matching. jcss, vol.26, n.3, pp.280--294, (1983).	
\item[7]  	A. Apostolico and R. Giancarlo. The Boyer-Moore-Galil string searching strategies revisited. SIAM J. Comput., vol.15, n.1, pp.98--105, (1986).	
\item[8]  	R. F. Zhu and T. Takaoka. On improving the average case of the Boyer-Moore string matching algorithm. J. Inform. Process., vol.10, n.3, pp.173--177, (1987).	
\item[9]  	R. M. Karp and M. O. Rabin. Efficient randomized pattern-matching algorithms. ibmjrd, vol.31, n.2, pp.249--260, (1987).	
\item[10]  	D. M. Sunday. A very fast substring search algorithm. Commun. ACM, vol.33, n.8, pp.132--142, (1990).
item[11]  	A. Apostolico and M. Crochemore. Optimal canonization of all substrings of a string. Inf. Comput., vol.95, n.1, pp.76--95, (1991).		
\item[12]  	L. Colussi. Correctness and efficiency of the pattern matching algorithms. Inf. Comput., vol.95, n.2, pp.225--251, (1991).	
\item[13]  	M. Crochemore and D. Perrin. Two-way string-matching. J. Assoc. Comput. Mach., vol.38, n.3, pp.651--675, (1991).
\item[14]  	A. Hume and D. M. Sunday. Fast string searching. Softw. Pract. Exp., vol.21, n.11, pp.1221--1248, (1991).	
\item[15]  	P. D. Smith. Experiments with a very fast substring search algorithm. Softw. Pract. Exp., vol.21, n.10, pp.1065--1074, (1991).	
\item[16]  	R. Baeza-Yates and G. H. Gonnet. A new approach to text searching. Commun. ACM, vol.35, n.10, pp.74--82, ACM, New York, NY, USA, (1992).	
\item[17] R. Baeza-Yates. String searching algorithms. In W. Frakes and R. Baeza-Yates, editors, Information Retrieval: Algorithms and Data Structures, chapter 10, pages 219--240. Prentice-Hall, 1992.
\item[18]  	Z. Galil and R. Giancarlo. On the exact complexity of string matching: upper bounds. SIAM J. Comput., vol.21, n.3, pp.407--437, (1992).	
\item[19]  	T. Lecroq. A variation on the Boyer-Moore algorithm. Theor. Comput. Sci., vol.92, n.1, pp.119--144, (1992).	
\item[20]  	T. Raita. Tuning the Boyer-Moore-Horspool string searching algorithm. Softw. Pract. Exp., vol.22, n.10, pp.879--884, (1992).	
\item[21]   WU, S., MANBER, U., 1992, Fast text searching allowing errors, Commun. ACM. 35(10):83-91.
\item[22]  	L. Colussi. Fastest pattern matching in strings. J. Algorithms, vol.16, n.2, pp.163--189, (1994).
\item[23]  	M. Crochemore and A. Czumaj and L. Gcasieniec and S. Jarominek and T. Lecroq and W. Plandowski and W. Rytter. Speeding up two string matching algorithms. Algorithmica, vol.12, n.4/5, pp.247--267, (1994).		
\item[24] CROCHEMORE, M., LECROQ, T., 1997, Tight bounds on the complexity of the Apostolico-Giancarlo algorithm, Information Processing Letters 63(4):195-203.
\item[25] D. Gusfield, Algorithms on Strings, Trees and Sequences, Cambridge University Press (1997).
\item[26]  	C. Charras and T. Lecroq and J. D. Pehoushek. A Very Fast String Matching Algorithm for Small Alphabets and Long Patterns. Proceedings of the 9th Annual Symposium on Combinatorial Pattern Matching, Lecture Notes in Computer Science, n.1448, pp.55--64, Springer-Verlag, Berlin, rutgers, (1998).	
\item[27]  	T. Berry and S. Ravindran. A fast string matching algorithm and experimental results. Proceedings of the Prague Stringology Club Workshop '99, pp.16--28, ctu, (1999).	
\item[28]  C. Charras and T. Lecroq, Handbook of Exact String Matching Algorithms, Kings College Publications, (2004).
\item[29] M. Crochemore, C. Hancart and T. Lecroq, Algorithms on Strings, Université de Rouen, Cambridge University Press, 
(2007).
\item[30] S. Faro and T. Lecroq. The Exact Online String Matching Problem: a Review of the Most Recent Results, 
ACM Computing Surveys (CSUR), 2013. 
\item[31]  S. Faro and T. Lecroq, A String Matching Algorithms Research Tool, University of Rouen $(URL: http://www-igm.univ-mlv.fr/~lecroq/lec_en.html)$. 
\end{hangref} 
\end{document}